\documentclass[runningheads]{llncs}
\usepackage{ifthen}
\newboolean{isarxivversion} %
\setboolean{isarxivversion}{false} %
\usepackage[T1]{fontenc}

\usepackage{numprint}
\usepackage{balance}
\usepackage{cite}
\usepackage{amsmath,amssymb,amsfonts}
\usepackage{graphicx}
\usepackage{textcomp}
\usepackage{xcolor}
\usepackage{tcolorbox}

\usepackage{siunitx}

\usepackage{algorithm}
\usepackage[noend]{algorithmic}
\newcommand\algorithmicprocedure{\textbf{procedure}}
\newcommand{\algorithmicendprocedure}{\algorithmicend\ \algorithmicprocedure}
\usepackage{numprint}
\npdecimalsign{.}
\makeatletter
\newcommand\PROCEDURE[3][default]{%
  \ALC@it
  \algorithmicprocedure\ \textsc{#2}(#3)%
  \ALC@com{#1}%
  \begin{ALC@prc}%
}
\newcommand\ENDPROCEDURE{%
  \end{ALC@prc}%
  \ifthenelse{\boolean{ALC@noend}}{}{%
    \ALC@it\algorithmicendprocedure
  }%
}
\newenvironment{ALC@prc}{\begin{ALC@g}}{\end{ALC@g}}
\makeatother
\usepackage{graphicx}
\usepackage{amsmath}
\usepackage{amsfonts}
\usepackage{hyperref}

\usepackage{tabularx} %
\usepackage{booktabs} %

\usepackage{longtable} %
\usepackage{ltablex} %
\usepackage{tabularray}
\UseTblrLibrary{siunitx}

\usepackage{subcaption}

\usepackage{color}
\usepackage{xcolor}
\usepackage{tikz}

\definecolor{svgRed}{HTML}{B80500}   %
\definecolor{svgGreen}{HTML}{00AF7F} %
\definecolor{dark}{HTML}{000000} %
\definecolor{darkturquoise}{RGB}{0, 128, 128}

\definecolor{NewViolet}{RGB}{112, 15, 186}

\usepackage[normalem]{ulem}
\usepackage{soul}

\usepackage{xspace}
\newcommand{\algoFont}[1]{{\textsc{#1}}\xspace}

\newcommand{\competitorOUR}{\algoFont{HLM:C (GPU)}}
\newcommand{\competitorOURKokkosGPU}{\algoFont{HLM:K (GPU)}}
\newcommand{\competitorOURKokkosOneC}{\algoFont{HLM:K (1C)}}
\newcommand{\competitorOURKokkosSixteenC}{\algoFont{HLM:K (16C)}}

\newcommand{\competitorStreaming}{\algoFont{Stack Streaming}}
\newcommand{\competitorBIRN}{\algoFont{LM (MPI)}}
\newcommand{\competitorGreedy}{\algoFont{Greedy}}
\newcommand{\competitorSUMAC}{\algoFont{SuMaC}}

\newcommand{\competitorGBBS}{\algoFont{gbbs}}

\usepackage{svg}
\begin{document}
\title{Efficient Parallel Algorithms for Hypergraph Matching}
\author{Henrik Reinstädtler\inst{1} \and
    Christian Schulz\inst{1} \and
    Nodari Sitchinava\inst{2}\and
    Fabian Walliser \inst{1}}
\authorrunning{H. Reinstädtler et al.}
\institute{Heidelberg University, Heidelberg, Germany\\
    \email{\{henrik.reinstaedtler,christian.schulz,fabian.walliser\}\\@informatik.uni-heidelberg.de}
    \and
    University of Hawaii at Manoa, Manoa, United States\\
    \email{nodari@hawaii.edu}
}
\maketitle              %
\begin{abstract}
    We present efficient parallel algorithms for computing maximal matchings in hypergraphs. Our algorithm finds locally maximal edges in the hypergraph and adds them in parallel to the matching. In the CRCW PRAM models our algorithms achieve $O(\log{\log{\Delta}}\log{m})$ time with $O(\kappa\log {m})$ work w.h.p.
    where $m$ is the number of hyperedges,
    and $\kappa$ is the sum of all vertex degrees. The CREW PRAM model algorithm has a running time of $O((\log{\Delta}+\log{d})\log{m})$ and requires $O(\kappa \log {m})$ work w.h.p. It can be implemented work-optimal with $O(\kappa)$ work in $O((\log{m}+\log{n})\log{m})$ time.
    We prove a~$1/d$-approximation guarantee for our algorithms.
    We evaluate our algorithms experimentally by implementing and running the proposed algorithms on the GPU using CUDA and Kokkos. Our experimental evaluation demonstrates the practical efficiency of our approach on real-world hypergraph instances, yielding a speed up of up to 76 times compared to a single-core CPU algorithm.
\end{abstract}
\section{Introduction}
We introduce an algorithm to compute hypergraph matchings on the GPU and test it on a wide variety of instances.
A hypergraph is a natural extension where more than two vertices can be connected by a single hyperedge.
A matching in a hypergraph is a vertex-disjoint edge set, which  optimizes either the number or the total weight of edges in it.
Its applications range from personnel scheduling~\cite{froger2015set} to combinatorial auctions~\cite{gottlob2013decomposing}.
It is known to be NP-complete~\cite{approxresult}, but can be approximated to a factor of $1/d$ by a simple greedy algorithm~\cite{Besser+:Greedy,Korte+:Greedy}, where $d$ is the maximum edge size.
The problem is well studied in the sequential~\cite{dufosse2019effective}, semi-streaming~\cite{reinstadtler2025semi}, and distributed setting~\cite{hanguir2021distributed}.

There is limited work on solving NP-complete problems on GPUs, e.g.~independent set problems by Burtscher~et~al.~\cite{burtscher2018high}.
There are several algorithms for the related problem of matching in graphs on GPUs or in parallel~\cite{auer2012gpu,manne2014new,birn2013efficient}.
Most notably, Birn~et~al.~\cite{birn2013efficient} introduced a parallel local max strategy, in which edges are added if their weight exceeds that of all their neighbors.
\paragraph*{Our Contribution.}
We present an algorithm that finds locally maximal edges and computes matchings in hypergraphs on the GPU.
The  running time is $O(\log{\log{\Delta}}\log{m})$ with $O(\kappa \log{m})$ work on the CRCW model and $O((\log{\Delta} + \log{d}) \log{m})$ running time with $O(\kappa \log{m})$ work on the CREW PRAM model, where $m$ is the number of hyperedges, plus $n$ the number of vertices,~$\Delta$ is the maximum degree, $d$ is the maximum hyperedge size, and $\kappa$ is the sum of all vertex degrees.
Additionally, we provide a work-optimal CREW algorithm requiring $O(\kappa)$ work and $O((\log{m}+\log{n})\log{m})$ time  w.h.p.
Each algorithm has an approximation guarantee of $1/d$.  The algorithms are implemented in CUDA and Kokkos~\cite{KokkosEcosystem2021}, ensuring improved cross-architecture compatibility. We evaluate on a diverse set of instances using both GPU and CPU platforms, showing  a speed-up of up to \numprint{76} times relative to a single-core CPU implementation.
We show the that the algorithm scales well on GPU and CPU when implemented with Kokkos.
On graphs, our algorithm \competitorOUR is faster than the  multicore CPU implementation by Birn~et~al.~\cite{birn2013efficient} in half of the categories and \competitorGBBS~\cite{dhulipala2021theoretically} on all except one.
Compared to the latest GPU matching by Mandulak et al.~\cite{mandulak2024efficient}, our Kokkos-based implementation \competitorOURKokkosGPU it is up to~\numprint{27.16} times faster and wins all categories, while producing a relative matching quality of~\numprint{87.6} to~\numprint{98.2}~\%.

The rest of the paper is structured as follows. Section~\ref{sec:prelim} introduces basic concepts and notations, while Section~\ref{sec:relwork} summarizes the related work in our field.
The algorithms are presented in Section~\ref{sec:algos:hm} and details of the implementation are discussed in Section~\ref{sec:details}. We experimentally evaluate our algorithms and discuss the results in Section~\ref{sec:exp}. Finally, we conclude in Section~\ref{sec:conclude}.
\section{Preliminaries}\label{sec:prelim}
\subsection{Basic Concepts}
\subsubsection{Hypergraphs.}
A weighted undirected hypergraph $H=(V,E,\omega)$ consists of a set $V$ of $n$ vertices and a set $E$ of~$m$~hyperedges. Each hyperedge $e \subseteq V$ is a subset of vertices whose weight is defined by the weight function $\omega: E\to \mathbb{R}_{> 0}$. %
The \emph{size} of a hyperedge is the number of vertices it contains and is denoted by $\left|e\right|$. The maximum size of a hyperedge, also called the \emph{rank}, is defined by $d := \max_{e \in E}\left|e\right|$. If all hyperedges have the same size $d$, the hypergraph is called $ d$-uniform.
Two edges are considered \emph{disjoint} if they do not share any common vertices and are \emph{neighbors} if they do share at least one vertex. An edge is considered \emph{locally maximal} if it has a greater weight than all of its neighbors.
The set of hyperedges that contain a vertex $v$ is defined by $E[v]$. Thereby, $deg(v) = |E[v]|$ denotes the number of hyperedges a vertex $v$ is contained in and is called the \emph{degree} of $v$. The maximum degree is denoted by $\Delta := \max_{v \in V}deg(v)$. We assume that vertices with degree $0$ have been removed during the preprocessing phase. Therefore, \hbox{$\Delta > 0$} or else the hypergraph is empty.
Furthermore, $\kappa = \sum_{v \in V} deg(v) = \sum_{e \in E} \left|e\right|$ is defined as the sum over all vertex degrees, which is equal to the sum of all edge sizes. Note that $\kappa \leq \min\{\Delta n, d m\}$.
\subsubsection{Matching.} A subset of hyperedges $M \subseteq E$ is a \emph{matching}, if all hyperedges are pairwise disjoint. The weight of a matching $M$ is defined by $w(M):= \sum_{e\in M}w(e)$. A matching $M$ is called \emph{maximal} if no hyperedge can be added without violating the matching property. A \emph{maximum} matching possesses the largest possible weight of all matchings. If the weight of each hyperedge is exactly the same, the problem is referred to as \emph{maximum cardinality~matching}, maximizing $|M|$. Finding a cardinality matching or the more general weighted hypergraph matching is  NP-hard~\cite{approxresult}.
This is in contrast to ordinary graphs, for which the problem can be solved in polynomial time~\cite{edmonds_1965}.
\subsubsection{Line Graph.}
Given a hypergraph $H = (V, E)$, the line graph $L(H) = (V', E')$ is defined by $V' := E$, and $\{e_1, e_2\} \in E'$ if $e_1 \cap e_2 \neq \emptyset$, i.e.~there is an edge in $L(H)$ if the two corresponding hyperedges share at least one vertex.
\subsubsection{Random Number Generation.}
Since random number generation itself is a subject of ongoing research, we adopt the common assumption that generating a random number takes constant time, i.e., $O(1)$ complexity.

\subsection{Computational Models}

The PRAM model consists of processors (PE) with uniform access to shared global memory.
Computation proceeds in synchronous steps, where in each cycle every processor can perform a single machine operation, each operation taking unit time. Since an operation of a synchronous PRAM model may involve concurrent memory access by several processors, different variants of the PRAM model are distinguished depending on whether such access is permitted or whether exclusive access is enforced. The \emph{Concurrent Read Exclusive Write (CREW)} model allows for simultaneous read operations but restricts writes such that only one PE may write to a given memory location at any time. The \emph{Concurrent Read Concurrent Write (CRCW)} model allows both simultaneous read and write operations performed to the same shared memory location. For write conflicts in CRCW, several submodels are available: writes are permitted only if all PEs attempt to write the same value (\emph{Common CRCW}); an arbitrary value is chosen (\emph{Arbitrary CRCW}); the value of the PE with the smallest ID is selected (\emph{Priority CRCW}); or the written values are combined according to \hbox{a predefined operation (\emph{Combining CRCW}).}

The MapReduce~\cite{dean2008mapreduce} framework is one of the dominant paradigms for processing large-scale data sets and works by defining two functions: \emph{map} and \emph{reduce}. The map function processes the input in parallel, and the intermediate results are combined and simplified during the reduce phase.
When processing large instances, the external memory model~\cite{vitter1994algorithms} provides a two-level storage system, where secondary storage can be accessed randomly as well.

\section{Related Work}\label{sec:relwork}
\noindent In this section, we give an overview of the literature on~matchings in graphs and hypergraphs as well as related~problems.
\paragraph{Matching in Graphs.} One of the most famous results in computer science is the polynomial complexity of the matching problem in graphs by Edmonds~\cite{edmonds_1965}, which can be reduced to $O(mn)$ with the disjoint-set union data structure of Tarjan~\cite{DBLP:journals/jcss/GabowT85}. Based on the Edmonds algorithm, Schwing et al.~\cite{DBLP:conf/ipps/SchwingGS24} developed a parallel shared-memory algorithm for cardinality matching.
The first $O(m)$ $\frac{1}{2}$-approximation is due to Preis~\cite{preis1999linear} and uses locally dominant edges.
A simpler algorithm with the same approximation guarantee was developed by Drake and Hougardy~\cite{DH03a} using path growing (PGA) in linear time. Pettie and Sanders~\cite{PETTIE2004271} present a $\frac{2}{3}-\varepsilon$ approximation with expected running time of $\mathcal{O}(m\log{\frac{1}{\varepsilon}})$.
Birn et al.~\cite{birn2013efficient} engineered a parallel algorithm with a $\frac{1}{2}$-approximation guarantee and $O(\log^2{n})$ running time in the CREW PRAM model utilizing locally dominant edges.
Dhulipala~et~al.~\cite{dhulipala2021theoretically} implement in \competitorGBBS a parallel shared-memory algorithm by Blelloch~et~al.~\cite{10.1145/2312005.2312058} with linear work for matchings.
Manne and Halappanavar~\cite{manne2014new} propose the suitor algorithm, which computes the same matching as the greedy algorithm, but is easily parallelizable.
Naim~et~al.~\cite{naim2015optimizing} implement the suitor approach on GPUs.
Mandulak~et~al.~\cite{mandulak2024efficient} compute locally dominant (pointer-based) matchings on multi-GPU systems and achieve a speedup between $2$ and $45$ compared to single-GPU implementations.
\paragraph{Matching in Hypergraphs.}
For $d$-uniform hypergraphs, Hazan~et~al.~\cite{approxresult} show that the matching problem is not efficiently approximable within a factor of $\Omega(d / \ln(d))$ unless~$P = NP$.
In the unweighted setting, Fürer and Yu~\cite{fuererapprox} developed a polynomial-time algorithm  with an approximation ratio of $\frac{d + 1}{3} + \epsilon$.
In the weighted setting, Berman's \emph{SquareImp} algorithm~\cite{10.1007/3-540-44985-X_19} yields an approximation guarantee of $\frac{k + 1}{2} + \epsilon$. %
This was recently improved by Neuwohner, reducing the approximation factor to $\frac{k + \epsilon_k}{2}$~\cite{DBLP:journals/corr/abs-2106-03555} or $\frac{k + 1}{2} - 2 \cdot 10^{-4}$ for $k \geq 4$\cite{neuwohner2023passing}. We are unaware of any practical implementations of these algorithms.
Dufosse~et~al.~\cite{dufosse2019effective} introduce reduction rules for special $d$-partite, $d$-uniform hypergraphs using Karp-Sipser rules and a scaling argument to select more hyperedges.
Hanguir and Stein~\cite{hanguir2021distributed} present three algorithms in the MPC (massively parallel computation) model with a trade-off between the approximation guarantee, memory requirement per machine and number of rounds of communication.
Their best algorithm with respect to the approximation guarantee takes $O(\log{n})$ rounds and $O(nd)$ space per machine.
Balliu~et~al.~\cite{balliu2023distributed} study the maximal matching and the related independent set problem in hypergraphs in the LOCAL model. They show that with unlimited message size, any deterministic algorithm requires at least $\Omega(\min\{\Delta d, \log_{\Delta d}{n}\})$ rounds of communication and $\Omega(\min\{\Delta d, \log_{\Delta d}{\log {n}}\})$ rounds if the algorithm is randomized.
Ghaffari and Trygub~\cite{ghaffari2024parallel} show that computing hypergraph matchings in parallel can be done in $O(\log{m})$ time and $O(md\log{m})$ work.
Blelloch and Brady~\cite{blelloch2025parallel} present a work-efficient algorithm with $O(\kappa)$ work and $O(\log^2{m})$ time w.h.p.
We are unaware of any practical implementations for these two algorithms.
In the semi-streaming model, Reinstädtler~et~al.~\cite{reinstadtler2025semi} present two approximation algorithms with a guarantee of $1/d(1+\varepsilon)$ in the semi-streaming model.
Großmann~et~al.~\cite{grossmann2024engineering} introduced reductions and local search for the relaxed $b$-matching that are directly applicable to hypergraph matching and implement these sequentially.
\paragraph{Related Problems} The maximum hypergraph matching problem is equivalent to the set packing problem, which involves finding a disjoint subset from a collection of sets such that the sum of its weights is maximized. These problems can also be reduced to the maximum independent set problem (MIS) in the line graph.
One of the first parallel algorithms for computing MIS is due to Karp and Wigderson~\cite{10.1145/800057.808690}.
The best known is Luby's algorithm~\cite{luby1985simple} which runs in $O(\log^2{n})$ using $O(m)$ processors in the arbitrary CRCW PRAM model.
It assigns random weights to the vertices of the graph and finds dominant vertices in parallel.
More recently, the ECL algorithm was proposed by Burtscher~et~al.~\cite{burtscher2018high} for GPUs building upon Luby's algorithm using reverse vertex degree as weights.

\section{Algorithms for Hypergraph Matching}\label{sec:algos:hm}
We begin with a sequential greedy local maximum algorithm that identifies locally maximal edges in the hypergraph and adds them to the solution. We then develop parallel variants in the PRAM model. Our algorithms iteratively identify locally maximal edges based on weights and add them to the matching while removing conflicting edges. The parallel variants offer different trade-offs between running time and work complexity under varying memory access constraints, with the CRCW variant achieving better time bounds through concurrent writes, while the CREW variant requires additional coordination but maintains exclusive write access.

\subsection{Sequential Algorithm}\label{sequ:description}
The sequential local-max algorithm works in two steps per round. In the first step we iterate over every active vertex and store the highest weight edge it is part of. Then we check for every edge if it was marked by all its vertices.
If they agreed, we add the edge to the matching and remove its neighbors.
This process is repeated over several rounds until every edge has either been matched or removed.
Before every round, small random perturbations are added to the edge weights taking $O(1)$ time. This noise is essential to reduce the number of rounds and to prevent the algorithm from getting stuck when two edges have identical weights.

Instead of removing, we soft-delete the hyperedges being removed by marking them as~inactive in order to reduce the number of times remaining edges would need to be repacked.
The implementations in the several models differ only in details of these steps and are discussed in Section~\ref{sec:details}.

It is straightforward to see that this algorithm generates a maximal matching, since at most one edge per vertex is added to the matching, and this proceeds until each edge is either part of the matching or inactive.

\subsubsection{Approximation Guarantee.}
The following lemma gives an approximation guarantee for the resulting weight.
\begin{lemma}\label{lemma:approx_guarantee}
    In a weighted hypergraph $H$, the weight of any maximal matching~$M$ produced by a local max algorithm is at least $1 / d$ of the maximum matching, where $d$ denotes the maximum edge size.%
\end{lemma}
\begin{figure}[t]
    \centering
    \includegraphics{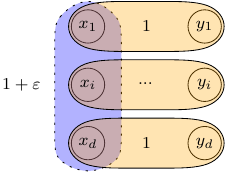}
    \caption{Visualization for the proof of the tight lower bound in Lemma~\ref{lemma:approx_tight_bound}. There are edges that each contain a pair of nodes $x_{i}$ and $y_{i}$ for all $i \in \{1, \dots, d\}$ with weight $1$. Additionally, all of the nodes $x_1$ to $x_{d}$ are also connected by an edge with weight $1 + \epsilon$ (\dotuline{\textcolor{blue}{blue}}).}
    \vspace{-0.6cm}
    \label{fig:tightness}
\end{figure}
\begin{proof}
    Let $\mathcal{M}$ be a maximum matching and $M$ the maximal matching of  $H$ found by our algorithm.
    Any edge in $\mathcal{M}$ is either in $M$ or shares a vertex with an edge in $M$. Otherwise, $M$ would not be maximal.

    For every of its up to $d$ vertices, an edge~$e$ in $M\setminus \mathcal{M}$ can be a neighbor of  an edge~$f$ in the given maximum matching~$\mathcal{M}$ that is not in $M$.
    However, $f$ must not be locally maximal; otherwise, it would be in $M$.
    So $\omega(e)>\omega(f)$ (for at least one of its neighbors in $M$) and $d\omega(e)>\sum_{f\in \mathcal{M}\setminus M\colon f\cap e\neq \emptyset}\omega(f)$.
    We rewrite the sum of weights by splitting the sum of $\mathcal{M}$ into $\mathcal{M}\cap M$ and $\mathcal{M}\setminus M$  in $\mathcal{M}$ and bound the latter by scaling the weight of~$M$~$d$~times.
    {\footnotesize\begin{align*}
        \sum_{f\in \mathcal{M}}\omega(f) & =\sum_{f\in \mathcal{M}\setminus M}\omega(f)+\sum_{f\in \mathcal{M}\cap M}\omega(f)
                                         & \leq \sum_{e\in M\setminus \mathcal{M}}d \omega(e)+\sum_{e\in \mathcal{M}\cap M}\omega(e) & \leq d\sum_{e\in M}\omega(e).
    \end{align*}}
\end{proof}

The following lemma demonstrates that the above bound is arbitrarily close to optimal, which is further illustrated in Figure~\ref{lemma:approx_guarantee}.
\begin{lemma}\label{lemma:approx_tight_bound}
    There exists a hypergraph $H$, such that the local-max algorithm cannot obtain approximation ratio better than  $1/(d(1+\epsilon))$ for any $\epsilon > 0$.
\end{lemma}

\begin{proof}
    Consider a hypergraph $H$ with $n=2d$ vertices shown in Figure~\ref{fig:tightness}. Each vertex is exactly connected to one other with weight $1$, resulting in $d$ edges. Add now an edge, which is connected to exactly one vertex of each pair and has weight $1 + \epsilon$ with $\epsilon > 0$. Applying the local max algorithm results in a weight of $1 + \epsilon$ for the matching, but the optimal value is $d$ when choosing all the edges with weight $1$.
\end{proof}

\subsection{Parallel Algorithms} \label{pram:description}
Starting from the sequential algorithm, we derive a parallel algorithm for weighted hypergraph matching by applying the same strategy using a Combining CRCW model. In the following, we present two refined variants of the model using either concurrent or exclusive writes that solve the unweighted (cardinality) hypergraph matching problem by randomizing the edge weights in each round. This design choice enables a more accurate estimation of the number of rounds (see Claim~\ref{thrm:rounds}). This is necessary since the analysis assumes a uniform distribution of weights for each round independently.

\paragraph{Sum-CRCW PRAM Algorithm.} \label{CRCW_PRAM_ALGORITHM}
In Algorithm~\ref{alg:para:hypermatching:crcw}, we present a SUM-CRCW PRAM algorithm.
In the SUM-CRCW model concurrent writes to the same variable are resolved as additions.
Each round of the algorithm consists of four sub-phases. At the beginning of each round, the first phase (Line~\ref{crcw:line:weights:begin}-\ref{crcw:line:weights:end}) reassigns random weights to the edges. In the second phase, each vertex identifies the heaviest incident edge. %
The third phase adds all locally maximal edges to the matching and marks their adjacent vertices as newly completed (Line~\ref{crcw:line:mark_and_newly_completed:begin}-\ref{crcw:line:mark_and_newly_completed:end}). In the final phase, all newly completed vertices as well as their adjacent edges are marked as inactive (Line~\ref{crcw:line:mark_inactive:begin}-\ref{crcw:line:mark_inactive:end}).
An edge is marked inactive by writing an arbitrary positive value and active edges can be identified by having zero as value in the status array of the edges.
These four sub-phases are repeated until all edges are either included in the matching or marked as inactive.

\begin{algorithm}[t]
    \caption{CRCW PRAM Algorithm for Hypergraph Matching}
    \label{alg:para:hypermatching:crcw}
    {\footnotesize
        \begin{algorithmic}[1]
            \PROCEDURE{CRCW-LocalMax}{$H=(V,E)$}
            \STATE $P\gets [0]\times \left|E\right|$ \COMMENT{Top edge id count per edge}
            \WHILE{\textbf{not} all edges marked inactive}
            \FOR{$e \in E$  \textbf{active in parallel}}\label{crcw:line:weights:begin}
            \STATE $w(e)\gets \mathrm{rnd()}$ \COMMENT{Initialize weights}
            \ENDFOR \label{crcw:line:weights:end}
            \STATE reset $P$
            \FOR{$v \in V$ \textbf{active in parallel}} \label{crcw:line:maximal_edges:begin}
            \STATE  Find $e:=\mathrm{argmax}_{e\in E(v), e\text{ active}}(w(e))$.
            \STATE \textbf{If} $T[v]\neq e$ \textbf{then}\quad  $T[v] \gets e$ and increase $P[e]$ \COMMENT{concurrent write}
            \ENDFOR\label{crcw:line:maximal_edges:end}
            \FOR{$e\in E$ \textbf{active in parallel}} \label{crcw:line:mark_and_newly_completed:begin}
            \STATE\textbf{if} $P[e]=\left| e\right|$ \textbf{then}\,  mark $e$ as in Matching and all $v \in e$ as newly completed\label{crcw:line:mark_and_newly_completed:end}
            \ENDFOR
            \FOR{$v\in V$ \textbf{active in parallel}}\label{crcw:line:mark_inactive:begin}
            \IF{$v$ marked newly completed}
            \STATE \textbf{for} $e \in E[v]$ active \textbf{do in parallel} mark $e$ as inactive \COMMENT{concurrent write}\label{crcw:line:set_edges_inactive}
            \ENDIF
            \STATE mark $v$ as inactive
            \ENDFOR \label{crcw:line:mark_inactive:end}
            \ENDWHILE
            \ENDPROCEDURE
        \end{algorithmic}}
\end{algorithm}

\textit{Running Time:} We divide the analysis of the running time into two main parts. First, we examine the time complexity of a single round. Then, leveraging the properties of the line graph, we establish an upper bound on the total number of rounds required.  This results in the following theorem.
\begin{theorem}
    Algorithm~\ref{alg:para:hypermatching:crcw} for cardinality matching runs in $O(\log{\log{\Delta}}\log{m})$ time and \hbox{$O(\kappa\log{m})$} work \hbox{w.h.p.}~in Sum-CRCW PRAM.
\end{theorem}
\begin{claim}\label{thrm:crcw:work_per_round}
    The time per round is bounded by $O(\log{\log{\Delta}})$ with work $O(\kappa)$ on the Sum-CRCW PRAM model, where $\kappa$ is the sum of  (vertex) degrees and $\Delta$ the maximum vertex degree in the hypergraph.
\end{claim}

\begin{proof}
    One round of Algorithm~\ref{alg:para:hypermatching:crcw} consists of four sub-phases. \\
    In the first subphase, using $m$ PEs, randomization can be done in parallel in $O(1)$ time, incurring $O(m)$ work.
    In the second subphase, we find the maximal adjacent edge in~$O(\log{\log{\Delta}})$ time using up to $\Delta/\log{\log{\Delta}}$ processors per vertex incurring in total~$O(\kappa)$ work, as described in~\cite{doi:10.1137/0204030}. Our naive implementation in the experiments finds this in~$\Delta$ steps using $O(n)$ processors.
    The subsequent part of the loop performs a concurrent write by incrementing $P[e]$, where $P$ is an array that stores the number of vertices for which an edge~$e$ is maximal. This loop requires $O(1)$ time while performing $O(\kappa)$ work.
    In the third subphase, using up to $\kappa$ PEs, marking vertices as newly completed can be done in $O(1)$ while using $O(\kappa)$ work. A concurrent write here is not needed since there is only one locally maximal edge per vertex.
    In the fourth subphase, using up to $\kappa$ PEs, the deactivation of edges may lead to concurrent writes and results in $O(1)$ time, incurring $O(\kappa)$ work.
    This results in an overall parallel running time of $O(\log{\log{\Delta}})$ using up to $n\Delta/\log{\log{\Delta}}$ PEs and a total amount of work that is bounded by $O(\kappa)$.
\end{proof}

The following claim shows that the number of rounds is only dependent on the number of hyperedges, completing the proof. %

\begin{claim}\label{thrm:rounds}
    Assume the edge weights are distributed uniformly w.h.p. and unique. For hypergraphs, the number of rounds required by the local max algorithm  is bounded by $O(\log{m})$ w.h.p. %
\end{claim}
\begin{proof}
    Consider the line graph $L=(V_L,E_L)$, such that there is one vertex in $V_L$ for every hyperedge in $H$ and two vertices are connected if the represented hyperedges share at least one vertex in the hypergraph.
    Note that the number of vertices $\left|V_L\right|$ is $m$.
    Our algorithm is now equivalent to Luby's algorithm~\cite{luby1985simple}~on~$L$. This algorithm identifies locally maximal vertices in the (line) graph and removes them and their neighbors.
    Luby's algorithm requires $\log{|V_L|}$ rounds w.h.p., assuming the edge weights are distinct.
    A hyperedge is selected if it has the highest weight out of all its neighbor edges in our algorithm, which also happens in Luby's algorithm, if a vertex has higher weight than its neighbors (aka conflicting hyperedges).
    Because $n'=m$, we can therefore conclude that there are $O(\log{m})$ rounds w.h.p..
\end{proof}

\paragraph{Memory Consumption.} Our algorithm requires the hypergraph to be present in memory. The representation of this is in $O(\kappa)$, the sum over all vertex degrees.  We additionally need to store the status per edge and node and $P$, which are all in~$O(m)$ or~$O(n)$, which  are in $O(\kappa)$.
\paragraph{CREW PRAM Algorithm.} \label{section:CREW}
The CRCW algorithm can be converted into a CREW algorithm by temporarily storing information about locally maximal edges and processing it in an additional step. The deactivation of hyperedges must also be carried out in further steps in order to avoid concurrent writes.
They are implemented using prefix-sum reductions which take $\log{}$ time and linear work, see Section~\ref{sec:details:crew}.
This results in the running time and work claimed in the following theorem. The proof can be found in Appendix~\ref{app:crew:proof}.

\begin{theorem}
    The CREW PRAM algorithm runs in $O((\log{\Delta} + \log{d}) \log{m})$ time and $O(\kappa\log{m})$ work \hbox{w.h.p.}
\end{theorem}

\subsubsection{Work-Optimal CREW PRAM Algorithm.}\label{sec:work_optimal_crew_algorithm}
The algorithms presented in previous sections are designed to achieve a low theoretical running time with $O((\log{\Delta} + \log{d}) \log{m})$, while resulting in a suboptimal work of $O(\kappa  \log{m})$. This section describes a version of this algorithm with a running time of $O((\log{m} + \log{n}) \log{m})$ which only incurs $O(\kappa)$ work. This demonstrates a trade-off between running time and the total work performed.

To achieve a work-optimal algorithm, in each round, all marked edges are removed from the hypergraph. Otherwise, the algorithm proceeds analogously to the CREW PRAM model from Section~\ref{section:CREW} until all edges to be removed are marked inactive and the round in the CREW PRAM model is finished. At the end of each such round, the hypergraph representation has to be updated to compute the input hypergraph for the next round, i.e.,~the respective hyperedges have to be removed. This can be done in $O(\log{m} + \log{n})$ time, which is used by the following theorem.

\begin{theorem}
    The work-optimal CREW PRAM algorithm runs in $O((\log{m} + \log{n}) \log{m})$ time and $O(\kappa)$ work \hbox{w.h.p.} %
\end{theorem}
\begin{proof}
    We show that the compactification of a hypergraph can be done in $O(\log{m}+\log{n})$ time with $O(\kappa)$ work after each round of the CREW algorithm:
    In order to remove marked vertices and edges, a parallel prefix sum is used. Therefore, using either $n$ or $m$ PEs, the prefix sums can be done in $O(\log{n} + \log{m})$ time and with $O(n + m)$ work total. Also, using up to $\kappa$ PEs, the data structures (see Section~\ref{sec:details:optimal}) can be rebuilt in $O(\log{\Delta} + \log{d})$ with $O(\kappa)$ work.
    In summary, this results in a running time of $O(\log{m} + \log{n})$ for one round of the algorithm, while the work complexity remains at $O(\kappa)$.
    According to Claim~\ref{thrm:rounds}, the number of rounds is w.h.p. bounded by $O(\log{m})$, which implies a total running time of $O((\log{m} + \log{n}) \log{m})$.
    The work and size of the hypergraph decrease geometrically following the proof of number of rounds (Theorem~\ref{thrm:rounds}) using Luby's~\cite{luby1985simple} algorithm, leading to an overall complexity of $O(\kappa)$ w.h.p.
\end{proof}
\paragraph{Extension to MapReduce.}
Each round in our CRCW algorithm can be transformed into several map-reduce steps.
Let $M$ be the size of the memory of each compute node.
Goodrich~et~al.~\cite{goodrich2011sorting} state that any Sum-CRCW algorithm,  that runs with $P$ processors and $N$ memory cells in $T$ time, can be simulated in $T\log_{M}{P}$ rounds with $T(N + P) \log_{M}{(N + P)}$ communication complexity.
We use up to \hbox{$N=P=n\Delta/\log{\log{\Delta}}$} cells/PEs in the parallel arg-max operation, so  $O(\log{\log{\Delta}}\log{m}\log_{M}{n\Delta/\log{\log{\Delta}}})$ rounds and $O(\log{\log{\Delta}}\log{m}(n\Delta/\log{\log{\Delta}})+ \log_{M}{(n\Delta/\log{\log{\Delta}})})$ communication complexity are needed. %

\subsubsection{Extension to External Memory.}
According to the geometrically decreasing input size, using the work-optimal CREW algorithm, the PRAM emulation techniques from Chiang~et~al.~\cite{DBLP:conf/soda/ChiangGGTVV95} allow for an implementation in the external memory~\cite{DBLP:journals/cacm/AggarwalV88} and cache-oblivious~\cite{DBLP:conf/focs/FrigoLPR99} models in $O((\log{m}+\log{n}) \mathrm{sort}({\kappa}))$ I/O operations, where $\mathrm{sort}(N)$ is defined as the number of operations for sorting $N$ elements using the parameters in the external memory model. %

\section{Implementation Details}\label{sec:details}
This section discusses the implementation details of our algorithms.
\subsection{Sequential Algorithm}\label{sec:details:sequential_algorithm}
\noindent

As the foundation for this approach, we propose a specialized data structure for representing hypergraphs. Accordingly, we utilize two arrays, $V_{id}$ and $V_{p}$, to store vertex-related information, and two arrays, $E_{id}$ and $E_{p}$, to store edge-related information. The indices of $V_{id}$ are interpreted as vertex identifiers. Moreover, the range $V_{id}[i]$ to $V_{id}[i+1]-1$ with $i \in \{0, \dots, n-1\}$ indicates a consecutive block of indices in $V_{p}$, each corresponding to an adjacent edge of vertex $i$. Furthermore, the values in $V_{p}$ contain the identifiers for these edges and refer to entries in $E_{id}$ themselves. The same structure is also applied to $E_{id}$ and $E_{p}$, where the range $E_{id}[j].first$ to $E_{id}[j+1].first - 1$ with $j \in \{0, \dots, m-1\}$ specifies a consecutive block of indices in $E_{p}$. Each of the corresponding values in $E_{p}$ identifies a vertex contained in edge $j$, and refers to an index of $V_{id}$. A second tuple entry $E_{id}[i].second$ contains any additional edge-metadata such as weights and flags. Similarly, a second tuple entry $V_{id}[i].second$ also contains additional vertex-metadata such as flags. For simplification, we refer to the combined representation $(V_{id}, V_{p})$ as $V$, and the pair $(E_{id}, E_{p})$ as $E$. Using these data structures allows efficient traversal in both directions by either iterating over edges adjacent to a certain vertex or over all vertices contained in an edge.

Before the algorithm initiates its computation, the weights are reconfigured by adding small noise values. For integer weights, a suitable approach is to generate floating-point random values in the interval $(0,1)$ which are mostly distinct w.h.p. %
In the unlikely event of conflicting equal weights, we use tie breakers by adding additional noise values.

Moreover, we allocate  two arrays, which store the IDs of the edge with the highest weight per vertex and the number of vertices for which an edge is maximal.
The remaining part of the algorithm is round-based and is performed until each edge is either part of the matching $M$ or marked inactive. Therefore, only active edges are considered at any step in the algorithm, which is assumed in the following description of such a round.
To start with, the edge with the locally maximal weight is identified for each vertex by a simple loop over all adjacent edges and the result is stored in an array $T$ accordingly. This is necessary as we are using this information in the following rounds to determine whether a vertex found a new maximal edge. %
Because if exactly that is the case (which it naturally is for the first round), we increase the pin counter $P[e]$ for an edge $e$ by one.
Therefore, to determine locally maximal edges, the algorithm compares the current pin count of an edge to its size. In this case, the edge is added to the matching and all of its adjacent edges, as well as contained vertices, \hbox{are marked inactive.}

\subsection{Parallel Algorithms}
In the following, we highlight the differences of Algorithm~\ref{alg:para:hypermatching:crcw} to the sequential model.
Specifically, we  transformed the outer \textbf{for}-loops into parallel loops. For iterations over the set of edges, the algorithm assumes the use of $m$ PEs.
Similarly, iterations over the set of vertices are parallelized using $n$ PEs.

\subsubsection{CRCW PRAM Algorithm}
To resolve the necessary processor communication, the increment operation for counting the nodes for which an edge is maximal must be performed atomically to prevent race conditions. Conflicts may also arise during the process of marking edges as inactive. However, since all processors attempt to set the same flag, these conflicts are benign and do not affect correctness.

\subsubsection{CREW PRAM Algorithm}\label{sec:details:crew}
This section provides a more detailed description of the CREW algorithm which consists of six sub-phases per round of computation.
In the first phase, random weights are assigned to the hyperedges, following the assumptions established in Section~\ref{sequ:description}.
In the second phase, every vertex identifies the adjacent edge with the biggest weight, which is temporarily stored in the array $T$. This information is then used in the third phase to determine the locally maximal edges, mark them as newly completed, and add them to the matching. In the fourth phase, vertices adjacent to matched edges are marked as newly completed, followed by the fifth phase, where we use a parallel OR-reduction to mark edges which contain newly completed vertices as inactive. In the final sixth phase, newly completed vertices are also marked as inactive. This write is non-concurrent, because every vertex can only be part of one newly completed edge.
These phases are repeated until all edges are marked inactive or as in the matching.

\subsubsection{Work-Optimal CREW PRAM Algorithm}\label{sec:details:optimal}
In the following, we describe in detail the compaction of a hypergraph, which is necessary for the work-optimal CREW algorithm. %

When  marked edges are removed, the information about matched edges has to be stored in an array $M$ for collecting all selected edges before they are excluded from the data structure. Since edges are actually removed and not just marked inactive, the \textit{active} constraint can be dropped for all loops. Otherwise, the algorithm proceeds analogously to the CREW PRAM model from Section~\ref{section:CREW} until all edges to be removed are marked inactive and the round in the CREW PRAM model is finished. At the end of such a round, the compaction of the hypergraph representation is initialized to prepare for the next round.

This is done by removing all marked vertices and edges from $V_{id}$ and $E_{id}$, respectively, and updating the pointers in $V_{p}$ and $E_{p}$. To inquire the new position for each vertex $V[i]$, we compute the number of entries $V[j], j < i$, marked inactive. Like before, we can perform parallel (exclusive) prefix sums on the inactive flags in $V_{id}$, which are stored in $P_V$. Furthermore, we declare a temporary storage for the compacted graph $V^\prime$ and copy the new degree of each vertex to the new positions, which are calculated by running the same scheme as for parallel prefix sums with $MAX$ as the associative operator on its neighbors. To now compute the new indices for $V^\prime_{p}$ in $V^\prime_{id}$, we perform another parallel (exclusive) prefix sum on $V^\prime_{id}$. Since our data structure for the hypergraph representation is symmetric, the same is applied to the edge set $E$.
To adjust the pointing variables in $V^\prime_{p}$ to the new positions, the algorithm iterates in parallel over the active neighbors of all vertices by taking the indices from $V_{id}$ and reading the old edge IDs from $V_{p}$. Since we already computed the new positions for active edge IDs with the prefix sum $P_E$, the value is taken from there and then written to the new index in $V^\prime_{p}$, which is derived from the already finished array $V^\prime_{id}$. Once again, we apply the same procedure in the case of the edges to $E^\prime_{p}$.
Moreover, we copy the compacted results in $V^\prime$ and $E^\prime$ to the old hypergraph representation $H = (V, E)$, which finalizes the preparation for the next round.

\section{Experimental Evaluation}\label{sec:exp}
This section reports experimental results on computing  maximal (cardinality and weighted) weighted hypergraph and unweighted graph matchings.
We compare our algorithms with the state-of-the-art implementations for graphs on GPUs \cite{mandulak2024efficient} and the CPU~\cite{birn2013efficient,dhulipala2021theoretically}.
For hypergraph matching, there are only CPU implementations available for comparison.
Therefore, we can only compare with the streaming algorithms by \cite{reinstadtler2025semi} using the CPU.
The implementations are benchmarked across a diverse set of hypergraph and graph instances.

Section~\ref{sec:exp:hypergraph} covers  experiments on randomly weighted hypergraphs, while Section~\ref{sec:exp:graph} contains the experiments for maximal cardinality matchings on graphs.

\subsection{Methodology}
We used a machine equipped with a NVIDIA RTX 4090 and  a 16-core Intel Xeon w5-3435X running at 3.10 Ghz having a L3 cache of 45 MB and 128 GB~RAM.

Each algorithm was executed three times per instance, with the total set of executions randomly ordered.
No experiment was scheduled in parallel.
The performance metric for each instance is computed using the arithmetic mean, and the geometric mean is applied when averaging across instances, thereby ensuring that each instance contributes comparably to the final result. While theoretical models such as PRAM assume an unlimited number of processors, in practice, a considerably smaller number of threads is launched and dynamically scheduled by the framework (Kokkos or CUDA) onto the limited physical processors available. Accordingly, the algorithms were configured to utilize the maximum number of cores available on the hardware. Since the hypergraph is assumed to be preloaded for a sequence of algorithm runs within a larger computation, the runtime for reading the hypergraph into the data structure, and transferring data to the GPU is excluded from the performance comparison.
The generation of new random noise per round is included in the running times.

\paragraph{Instances.}
The set of hypergraphs was collated from a set of instances created by Gottesbüren~et~al.~\cite{gottesburen2023}. Specifically, we use 90 of 94 instances from the large-sized \textsf{$L_{HG}$} set. These originate from the \emph{SAT competition 2014}~\cite{sat2014}, more general matrices from the \emph{SuiteSparse collection}~\cite{davis2011}, and the \emph{DAC challenge}~\cite{viswanathan2012}. Four instances were too big to fit onto the RTX 4090 memory.
For weighted experiments, we assign random weights between 1 and 100 to the hyperedges.
We add the weighted wiki instance from Reinstädtler~et~al.~\cite{reinstadtler2025wiki}.
A set of~42 graphs was obtained from \emph{10th DIMACS Challenge}~\cite{DBLP:conf/dimacs/2012}.
This collection includes \emph{Delaunay instances}~\cite{DBLP:conf/ipps/HoltgreweSS10}, generated by randomly selecting $n = 2^{x}$ points in the unit square and computing their Delaunay triangulation. In addition, \emph{random geometric graphs} from the same source are considered, having $n = 2^{x}$ nodes and $\alpha n$ edges with $\alpha \in \{4, 16, 64\}$. The benchmark set further compromises real-world graphs such as street networks and \emph{co-author and citation networks}~\cite{DBLP:conf/alenex/GeisbergerSS08}, which represent social networks. Finally, graphs from the \emph{Florida Sparse Matrix Collection}~\cite{davis2011}~are included.
\paragraph{Compared Algorithms.}
We implemented the algorithms in CREW and CRCW fashion in CUDA and Kokkos.
They add random noise, sampled uniformly from the interval~$[0,100]$, to the edges between each round.
In this section, we  compare our CRCW implementations \competitorOUR using an XORWOW generator~\cite{Marsaglia_XORWOW} and the Kokkos implementations \competitorOURKokkosGPU using GPU, one core~(\competitorOURKokkosOneC) and 16 cores (\competitorOURKokkosSixteenC).
On the Kokkos plattform we use a XORSHIFT random generator~\cite{vigna2016experimentalexplorationmarsagliasxorshift}.
The random number generator can be  switched to a simpler Park-Miller~\cite{DBLP:journals/cacm/ParkM88} or any other~generator.

We compare with the parallel \competitorBIRN algorithm for graphs by Birn et al.~\cite{birn2013efficient}, which is implemented with MPI and kindly provided by the authors. As additional CPU implementation we test against ~\competitorGBBS~\cite{dhulipala2021theoretically}.
Moreover, we compare it with \competitorSUMAC by Mandulak~et~al.~\cite{mandulak2024efficient} which can run on multiple GPUs.
We configured it to run on a single GPU and measured only the raw computation time, excluding initialization costs.
For experiments on hypergraphs, a single-core CPU  \competitorGreedy implementation and a streaming \competitorStreaming  by Reinstädtler~et~al.~\cite{reinstadtler2025semi} is used as a baseline. The streaming approach has a approximation guarantee of $1/(d+\varepsilon)$, we set $\varepsilon=0$ to match the greedy bound. It has a better running time than \competitorGreedy.

All GPU implementations were compiled with nvcc and CUDA 12.8.
For the experiments with \competitorBIRN we used OpenMPI 5.0.7 with gcc 13.0.3.
\subsection{Experiments on Hypergraphs}\label{sec:exp:hypergraph}
\begin{figure}[t]
    \vspace{-0.4cm}
    \centering
    \includegraphics{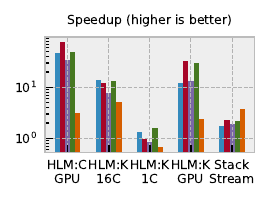}
    \includegraphics{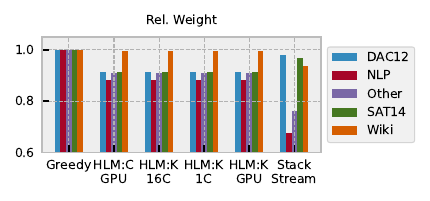}\vspace{-0.4cm}
    \caption{ Results on hypergraphs. Average speedup over \competitorGreedy (left) and relative quality (to the best result) per hypergraph category for our algorithms and \competitorStreaming~\cite{reinstadtler2025semi}.}
    \label{fig:res:hypergraphs}
    \vspace{-0.4cm}
\end{figure}
In Figure~\ref{fig:res:hypergraphs} we show the results of our experiments on hypergraphs against the competitors from~\cite{reinstadtler2025semi}.

\paragraph{Speedup over \competitorGreedy.} On the left side we show the speedup over the \competitorGreedy algorithm, that runs on the CPU on a single core and sorts all edges in-memory.
The \competitorStreaming is between 1.74 and 3.70 times faster than this.
Our one-core run \competitorOURKokkosOneC is slower on three of the categories (up to a factor of 1.5 times slower), but already on the \textsf{DAC12} and \textsf{SAT14} up to 1.61 times faster.
When adding more cores (\competitorOURKokkosSixteenC) the speedup is 13 times.
This shows how well our algorithm scales with respect to the number of cores.
Our Kokkos based GPU implementation  is up to 32 times faster than the baseline.
However, in the Wiki instance it only has a speed up of \numprint{2.38}, which is lower than the speed up of the 16 core CPU implementation.
The weights of the \textsf{Wiki} instance are not uniformly distributed.
Thus, all of our algorithms require 16 rounds to complete this instance instead of 5-7 on the other instances explaining the relative low speed up.
The best results with respect to running time, yielding a speedup between 3.20 and 76.6 times, are returned by our CUDA based implementation (\competitorOUR).
The main difference between the two implementations is the random number generator~used.

\vspace{-0.4cm}
\paragraph{Quality.} Our algorithms reach between \numprint{88.13} and \numprint{99.72}~\% of the weight of the \competitorGreedy matchings. The difference, caused by different random generators, are negligible.
The streaming algorithm (\competitorStreaming) returns only \numprint{67.51}~\% on the \textsf{NLP} and \numprint{76.14}~\% on the \textsf{Other} instances.
The solution quality on the \textsf{Wiki} instance is also worse than of our algorithm, even though it is the only category where the streaming algorithm was faster.
\subsection{Experiments on Graphs}\label{sec:exp:graph}
We conduct experiments on graphs to show that our approach is also feasible and efficient on graphs.
Figure~\ref{fig:res:graphs} shows the speedup  over \competitorBIRN~by  Birn~et~al.~\cite{birn2013efficient} and relative size of \competitorGBBS~\cite{dhulipala2021theoretically}, \competitorSUMAC~\cite{mandulak2024efficient} and our algorithms.
\competitorBIRN uses all 16 physical cores of the host via OpenMPI.
Our CPU implementation using Kokkos is done via OpenMP.
\competitorSUMAC is a GPU implementation of a local-dominant algorithm using a suitor technique. It uses the vertex id as tie breaker, while we use uniform edge weight and added random noise in the interval $[0,100]$. \competitorBIRN also uses random weights.
\begin{figure}[t]
    \centering
    \includegraphics{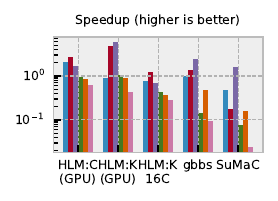}\hspace{-0.25cm}
    \includegraphics{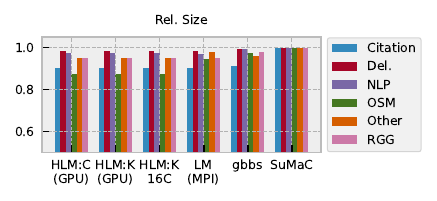}
    \vspace{-0.4cm}
    \caption{Experimental results on graphs. Speedup over \competitorBIRN by \cite{birn2013efficient} (left) and relative matching size  per graph category.}
    \vspace{-0.4cm}
    \label{fig:res:graphs}
\end{figure}

\paragraph{Speedup.} Even though the data structure for accessing hypergraphs is more complex, our parallel CPU implementation (\competitorOURKokkosSixteenC) performs \numprint{1.24} times faster on the \textsf{Delaunay} instances.
On the worst-case (\textsf{RGG}) graphs our implementation is \numprint{3.42} slower on the CPU. \competitorGBBS performs better than our  \competitorOURKokkosSixteenC on four of the categories, while paying performance penalties on \textsf{OSM} and \textsf{RGG}.
Our \competitorOUR implementation is faster in all except one category.

When looking at GPU implementations, we can see that our CUDA variant (\competitorOUR) and Kokkos (\competitorOURKokkosGPU) perform better.
\competitorOUR is faster on Citation (2.06 times), Delaunay (2.75) and NLP (1.70) than \competitorBIRN, which runs on 16 cores on the CPU.
On \textsf{OSM}, \textsf{Other} and \textsf{RGG} \competitorOUR is between  1.08 and 1.61 times slower than \competitorBIRN.
Again caused by the different random generators, the results to \competitorOURKokkosGPU differ and
\competitorOURKokkosGPU is more than twice as fast on \textsf{NLP} graphs. On \textsf{RGG}, \textsf{Citation} networks it is considerably slower than \competitorOUR.

When only comparing GPU algorithms, \competitorOUR   is up to \numprint{15.28} and \competitorOURKokkosGPU up to \numprint{27.16} times faster than \competitorSUMAC on \textsf{Delaunay} graphs.
Only for \competitorOUR the \textsf{NLP} instances from numerics are challenging, here our algorithm is only \numprint{1.04} times faster. These graphs are derived from matrices with~3 distinct bands of non-zeros, making them ideal for \competitorSUMAC using the edge~ids.

\vspace{-0.2cm}
\paragraph{Quality.} Our algorithms reach 87.6 to 98.2~\% of the solution size of \competitorSUMAC on unweighted graphs, while \competitorBIRN~\cite{birn2013efficient} reaches 90 to 97.8~\%, and \competitorGBBS~\cite{dhulipala2021theoretically}~\numprint{91} to \numprint{99.2}~\%.
The strict order given by the ids, used in \competitorSUMAC, seems to produce better matches. We could use these as edge weights as well.

\vspace{-0.4cm}
\section{Conclusion}\label{sec:conclude}
\vspace{-0.4cm}
In this paper, we presented a novel family of algorithms to compute matchings in hypergraphs in a efficient parallel fashion.

The algorithms find locally maximal hyperedges in parallel until all edges are matched or removed.
We show a running time of $O(\log{\log{\Delta}}\log{m})$ with~\hbox{$O(\kappa \log{m})$} work in the CRCW model and provide extensions to CREW, a work optimal version, MapReduce and external memory.
In experiments, our algorithms outperform CPU single-core implementations w.r.t. running time up to a factor of \numprint{76.6} and reaching a quality of 88.13~\% on weighted hypergraphs. In the graph cardinality case, our algorithm \competitorOUR is faster than  \competitorSUMAC on all, \competitorGBBS~\cite{dhulipala2021theoretically} except for one  and compared to~\competitorBIRN~\cite{birn2013efficient} on half of six  categories.

Our algorithm can be easily extended to multi-GPU setups, by partitioning, running it on the interior edges in parallel and handeling cross-edges on a single GPU.
Other future avenues of work include devising, a parallel local search or expanding to the relaxed $b$-matching problem, where multiple edges can be matched per~vertex.
\vspace{-0.3cm}
\begin{credits}
    \footnotesize
    \subsubsection{\ackname} We acknowledge support by DFG grant \hbox{SCHU 2567/8-1}, National Science Foundation grant \hbox{CCF-2432018}.
    \vspace{-0.4cm}
    \subsubsection{\discintname}
    The authors have no competing interests to declare that are
    relevant to the content of this article.
\end{credits}
\bibliographystyle{customttt}%
\vspace{-0.5cm}
\bibliography{lncs-clean}

\clearpage
\appendix
\section{Work per Round on the CREW model}
\label{app:crew:proof}
\begin{theorem}
    The CREW PRAM algorithm runs in $O((\log{\Delta} + \log{d}) \log{m})$ time and $O((\kappa ) \log{m})$ work \hbox{w.h.p.}
\end{theorem}
\begin{proof}
    One round of Algorithm~\ref{alg:para:hypermatching:crew} consists of six sub-phases. \\
    In the first subphase, using $m$ PEs, randomization can be done in parallel in $O(1)$ time, incurring $O(m)$ work. \\
    In the second subphase, using up to  $\kappa$ PEs, the locally maximal edge for each vertex is found using the same technique as for parallel prefix sums with MAX as the associative operator on each subset of edges that all contain a single vertex. Since $V$ has size $O(\kappa)$, this step takes $O(\log{\Delta})$ time and $O(\kappa)$ work. \\
    In the third subphase, using up to $\kappa$ PEs, counting all vertices for which an edge is maximal in line~\ref{crew:line:mark_matching:sum} can be performed using prefix sums with $\text{SUM}$ as the associative operator. Since $E$ has size $O(\kappa)$, this step takes $O(\log{d})$ while using $O(\kappa)$ work. \\
    In the fourth subphase, using up to $\kappa$ PEs, marking vertices as newly completed takes $O(1)$ time and $O(\kappa)$ work. Since for every vertex, at most one
    edge this vertex is contained in can be marked newly completed, this step involves only exclusive writes. \\
    In the fifth subphase, using up to $\kappa$ PEs, marking edges as inactive in parallel requires a logarithmic OR-reduction to prevent concurrent writes. This is achieved by using parallel prefix sums with OR as the associative operator on each subset of vertices that are contained in a single edge. This step takes $O(\log{d})$ time and $O(\kappa)$ work.
    In the sixth subphase, using up to $n$ PEs, marking vertices as inactive takes $O(1)$ time and $O(n)$ work.
    This results in an overall parallel running time of $O(\log{\Delta} + \log{d})$ while the total amount of work is limited by~$O(\kappa)$.
\end{proof}
\begin{algorithm}[b]
    \caption{CREW PRAM Algorithm for Hypergraph Matching}
    \label{alg:para:hypermatching:crew}
    \begin{algorithmic}[1]
        \PROCEDURE{CREW-LocalMax}{$H=(V,E)$}
        \STATE $T \gets [\bot]\times \left|V\right|$ \COMMENT{Top edge id for every vertex}
        \WHILE{\textbf{not} all edges marked inactive}
        \FOR{$e \in E$  \textbf{active in parallel}} \label{crew:line:weights:begin}
        \STATE $w(e)\gets \mathrm{rnd()}$ \COMMENT{Initialize weights}
        \ENDFOR \label{crew:line:weights:end}
        \FOR{$v \in V$ \textbf{active in parallel}}\label{crew:line:maximal_edges:begin}
        \STATE $e:=\mathrm{argmax}_{e \in E(v), e\text{ active}}(w(e))$ \label{line:argmax}
        \IF{$T[v]\neq e$}
        \STATE $T[v] \gets e$\label{crew:line:maximal_edges:end}
        \ENDIF
        \ENDFOR
        \FOR{$e\in E$ \textbf{active in parallel}}\label{crew:line:mark_matching:begin}
        \STATE $P_e = \sum_{v \in e}\begin{cases}
                1 & \text{if } T[v]=e \\
                0 & \text{else}
            \end{cases} $ \label{crew:line:mark_matching:sum}
        \IF{$P_e=\left| e\right|$}
        \STATE mark $e$ as in Matching
        \STATE mark $e$ as newly completed
        \ENDIF
        \ENDFOR \label{crew:line:mark_matching:end}
        \FOR{$v\in V$ \textbf{active in parallel}}\label{crew:line:mark_newly_completed:begin}
        \FOR{$e \in E[v] \textbf{ in parallel}$}
        \IF{$e$ is newly completed}
        \STATE mark $v$ as newly completed
        \ENDIF
        \ENDFOR
        \ENDFOR \label{crew:line:mark_newly_completed:end}

        \FOR{$e\in E$ \textbf{active in parallel}} \label{crew:line:mark_edge_inactive:begin}
        \STATE any\_marked = $\bigvee_{e \in E(v)} \text{active}(e)$
        \IF{any\_marked}
        \STATE mark $e$ as inactive
        \ENDIF
        \ENDFOR \label{crew:line:mark_edge_inactive:end}

        \FOR{$v\in V$ \textbf{active in parallel}} \label{crew:line:mark_vertex_inactive:begin}
        \IF{$v$ is newly completed}
        \STATE mark $v$ as inactive\label{line:mark:end}
        \ENDIF
        \ENDFOR \label{crew:line:mark_vertex_inactive:end}
        \ENDWHILE
        \ENDPROCEDURE
    \end{algorithmic}
\end{algorithm}
\clearpage
\section{Instance Statistics}
\subsubsection{Hypergraph instances}
\vspace{0.5cm}
\begin{longtblr}[
	caption = {Benchmark set of hypergraphs with $|V|$ representing the number of vertices and $|E|$  the number of hyperedges. $d$ is the maximum edge degree, $\Delta$ is the maximum node degree.},
	label = {table:hypergraph:inst:stat},
	]{colsep = 1mm,
	rowsep = 1mm,
	colspec = {lS[table-format=5.2]S[table-format=5.2]S[table-format=5.2]S[table-format=5.2]S[table-format=5.2]},rowhead =1,  row{even} = {gray9}, row{1} = {font=\scshape}, column{1}={font=\ttfamily}}
Name&{{{$\texttt{avg}|e|$}}} & {{{$|E|$}}} & {{{$d$}}} & {{{$\Delta$}}}& {{{$|V|$}}}\\
en\_wiki\_latest\_dec2024\_access & 3.82 & \numprint{7847749} & \numprint{228} & \numprint{1003480} & \numprint{293177} \\
Bump\_2911 & 43.87 & \numprint{2911419} & \numprint{195} & \numprint{195} & \numprint{2911419} \\
CurlCurl\_4 & 11.14 & \numprint{2380515} & \numprint{13} & \numprint{13} & \numprint{2380515} \\
Emilia\_923 & 44.42 & \numprint{923136} & \numprint{57} & \numprint{57} & \numprint{923136} \\
Flan\_1565 & 75.03 & \numprint{1564794} & \numprint{81} & \numprint{81} & \numprint{1564794} \\
Freescale1 & 5.52 & \numprint{3428755} & \numprint{27} & \numprint{25} & \numprint{3428755} \\
FullChip & 8.91 & \numprint{2987012} & \numprint{2312481} & \numprint{2312476} & \numprint{2987012} \\
Ga41As41H72 & 68.96 & \numprint{268096} & \numprint{702} & \numprint{702} & \numprint{268096} \\
Geo\_1438 & 43.92 & \numprint{1437960} & \numprint{57} & \numprint{57} & \numprint{1437960} \\
HV15R & 140.33 & \numprint{2017169} & \numprint{484} & \numprint{303} & \numprint{2017169} \\
RM07R & 98.16 & \numprint{381689} & \numprint{295} & \numprint{245} & \numprint{381689} \\
StocF-1465 & 14.34 & \numprint{1465137} & \numprint{189} & \numprint{189} & \numprint{1465137} \\
af\_shell10 & 34.93 & \numprint{1508065} & \numprint{35} & \numprint{35} & \numprint{1508065} \\
af\_shell6 & 34.84 & \numprint{504855} & \numprint{40} & \numprint{40} & \numprint{504855} \\
boneS10 & 60.63 & \numprint{914898} & \numprint{81} & \numprint{81} & \numprint{914898} \\
circuit5M & 10.71 & \numprint{5558326} & \numprint{1290501} & \numprint{1290501} & \numprint{5558326} \\
dac2012\_superblue11 & 3.28 & \numprint{935731} & \numprint{3602} & \numprint{1983} & \numprint{952507} \\
dac2012\_superblue12 & 3.69 & \numprint{1293436} & \numprint{8488} & \numprint{1006} & \numprint{1291931} \\
dac2012\_superblue14 & 3.31 & \numprint{619815} & \numprint{8871} & \numprint{1023} & \numprint{630802} \\
dac2012\_superblue16 & 3.27 & \numprint{697458} & \numprint{2367} & \numprint{1016} & \numprint{698339} \\
dac2012\_superblue19 & 3.35 & \numprint{511685} & \numprint{9155} & \numprint{1507} & \numprint{522482} \\
dac2012\_superblue2 & 3.26 & \numprint{990899} & \numprint{2048} & \numprint{1317} & \numprint{1010321} \\
dac2012\_superblue3 & 3.46 & \numprint{898001} & \numprint{2527} & \numprint{2245} & \numprint{917944} \\
dac2012\_superblue6 & 3.37 & \numprint{1006629} & \numprint{7601} & \numprint{1689} & \numprint{1011662} \\
dac2012\_superblue7 & 3.68 & \numprint{1340418} & \numprint{5310} & \numprint{849} & \numprint{1360217} \\
dac2012\_superblue9 & 3.48 & \numprint{833808} & \numprint{7772} & \numprint{1265} & \numprint{844332} \\
dgreen & 31.87 & \numprint{1200611} & \numprint{187} & \numprint{187} & \numprint{1200611} \\
dielFilterV2clx & 41.68 & \numprint{607232} & \numprint{158} & \numprint{158} & \numprint{607232} \\
dielFilterV2real & 41.94 & \numprint{1157456} & \numprint{110} & \numprint{110} & \numprint{1157456} \\
dielFilterV3clx & 78.22 & \numprint{420408} & \numprint{303} & \numprint{303} & \numprint{420408} \\
dielFilterV3real & 80.98 & \numprint{1102824} & \numprint{270} & \numprint{270} & \numprint{1102824} \\
fem\_hifreq\_circuit & 41.21 & \numprint{491100} & \numprint{110} & \numprint{110} & \numprint{491100} \\
gsm\_106857 & 36.91 & \numprint{589446} & \numprint{106} & \numprint{106} & \numprint{589446} \\
hugebubbles-00010 & 3.00 & \numprint{19458087} & \numprint{3} & \numprint{3} & \numprint{19458087} \\
inline\_1 & 73.09 & \numprint{503712} & \numprint{843} & \numprint{843} & \numprint{503712} \\
kmer\_P1a & 2.14 & \numprint{139353211} & \numprint{40} & \numprint{40} & \numprint{139353211} \\
ldoor & 48.86 & \numprint{952203} & \numprint{77} & \numprint{77} & \numprint{952203} \\
msdoor & 48.67 & \numprint{415863} & \numprint{77} & \numprint{77} & \numprint{415863} \\
nlpkkt160 & 27.50 & \numprint{8345600} & \numprint{28} & \numprint{28} & \numprint{8345600} \\
nlpkkt80 & 27.02 & \numprint{1062400} & \numprint{28} & \numprint{28} & \numprint{1062400} \\
rajat31 & 4.33 & \numprint{4690002} & \numprint{1252} & \numprint{1252} & \numprint{4690002} \\
sat14\_11pipe\_k.cnf.dual & 185.79 & \numprint{89315} & \numprint{57073} & \numprint{295} & \numprint{5584003} \\
sat14\_11pipe\_k.cnf & 2.97 & \numprint{5584003} & \numprint{295} & \numprint{28538} & \numprint{178630} \\
sat14\_11pipe\_k.cnf.primal & 2.97 & \numprint{5584003} & \numprint{295} & \numprint{57073} & \numprint{89315} \\
sat14\_9'.cnf.dual & 75.22 & \numprint{521179} & \numprint{220932} & \numprint{6689} & \numprint{13378617} \\
sat14\_9'.cnf & 2.93 & \numprint{13378617} & \numprint{6689} & \numprint{110469} & \numprint{1042358} \\
sat14\_9'.cnf.primal & 2.93 & \numprint{13378617} & \numprint{6689} & \numprint{220932} & \numprint{521179} \\
sat14'1.cnf.dual & 9.84 & \numprint{1621762} & \numprint{587} & \numprint{23} & \numprint{6356704} \\
sat14'1.cnf & 2.51 & \numprint{6356704} & \numprint{23} & \numprint{302} & \numprint{3195850} \\
sat14'1.cnf.primal & 2.51 & \numprint{6356704} & \numprint{23} & \numprint{587} & \numprint{1621762} \\
sat14'2.cnf.dual & 9.84 & \numprint{1621756} & \numprint{587} & \numprint{23} & \numprint{6356692} \\
sat14'2.cnf & 2.51 & \numprint{6356692} & \numprint{23} & \numprint{302} & \numprint{3195838} \\
sat14'2.cnf.primal & 2.51 & \numprint{6356692} & \numprint{23} & \numprint{587} & \numprint{1621756} \\
sat14\_SAT'3.cnf.dual & 9.84 & \numprint{1540071} & \numprint{587} & \numprint{23} & \numprint{6036162} \\
sat14\_SAT'3.cnf & 2.51 & \numprint{6036162} & \numprint{23} & \numprint{302} & \numprint{3034828} \\
sat14\_SAT'3.cnf.primal & 2.51 & \numprint{6036162} & \numprint{23} & \numprint{587} & \numprint{1540071} \\
sat14\_atco'50.dual & 9.94 & \numprint{1613160} & \numprint{1024} & \numprint{15} & \numprint{6429816} \\
sat14\_atco'50 & 2.50 & \numprint{6429816} & \numprint{15} & \numprint{513} & \numprint{3226318} \\
sat14\_atco'50.primal & 2.50 & \numprint{6429816} & \numprint{15} & \numprint{1024} & \numprint{1613160} \\
sat14\_atco'21.dual & 9.98 & \numprint{1484061} & \numprint{733} & \numprint{15} & \numprint{5933456} \\
sat14\_atco'21 & 2.50 & \numprint{5933456} & \numprint{15} & \numprint{368} & \numprint{2968120} \\
sat14\_atco'21.primal & 2.50 & \numprint{5933456} & \numprint{15} & \numprint{733} & \numprint{1484061} \\
sat14\'37'.dual & 39.77 & \numprint{559571} & \numprint{4259} & \numprint{1408} & \numprint{10223027} \\
sat14\'37' & 2.18 & \numprint{10223027} & \numprint{1408} & \numprint{2814} & \numprint{1119142} \\
sat14\'37'.primal & 2.18 & \numprint{10223027} & \numprint{1408} & \numprint{4259} & \numprint{559571} \\
sat14\_esawn\_uw3.debugged.cnf.dual & 10.89 & \numprint{13829558} & \numprint{40760} & \numprint{156} & \numprint{53616734} \\
sat14\_esawn\_uw3.debugged.cnf & 2.81 & \numprint{53616734} & \numprint{156} & \numprint{20380} & \numprint{27659037} \\
sat14\_esawn\_uw3.debugged.cnf.primal & 2.81 & \numprint{53616734} & \numprint{156} & \numprint{40760} & \numprint{13829558} \\
sat14\_sin.c'.dual & 11.29 & \numprint{7916927} & \numprint{155449} & \numprint{110} & \numprint{32591905} \\
sat14\_sin.c' & 2.74 & \numprint{32591905} & \numprint{110} & \numprint{77726} & \numprint{15833854} \\
sat14\_sin.c'.primal & 2.74 & \numprint{32591905} & \numprint{110} & \numprint{155449} & \numprint{7916927} \\
sat14\_sv-'\_3t\_true'.dual & 14.88 & \numprint{9580427} & \numprint{197796} & \numprint{1089} & \numprint{46604325} \\
sat14\_sv-'\_3t\_true' & 3.06 & \numprint{46604325} & \numprint{1089} & \numprint{193812} & \numprint{19160854} \\
sat14\_sv-'\_3t\_true'.primal & 3.06 & \numprint{46604325} & \numprint{1089} & \numprint{197796} & \numprint{9580427} \\
sat14\_sv-comp19'\_false'.dual & 14.70 & \numprint{8043315} & \numprint{12899} & \numprint{12897} & \numprint{39379566} \\
sat14\_sv-comp19'\_false' & 3.00 & \numprint{39379566} & \numprint{12897} & \numprint{12897} & \numprint{16086630} \\
sat14\_sv-comp19'\_false'.primal & 3.00 & \numprint{39379566} & \numprint{12897} & \numprint{12899} & \numprint{8043315} \\
sat14'b71.cnf.dual & 47.23 & \numprint{889302} & \numprint{228944} & \numprint{5391} & \numprint{14582074} \\
sat14'b71.cnf & 2.88 & \numprint{14582074} & \numprint{5391} & \numprint{114472} & \numprint{1778604} \\
sat14'b71.cnf.primal & 2.88 & \numprint{14582074} & \numprint{5391} & \numprint{228944} & \numprint{889302} \\
sat14\_zfcp-2.8-u2-nh.cnf.dual & 6.97 & \numprint{10950109} & \numprint{1777221} & \numprint{65} & \numprint{32697150} \\
sat14\_zfcp-2.8-u2-nh.cnf & 2.33 & \numprint{32697150} & \numprint{65} & \numprint{888611} & \numprint{21899848} \\
sat14\_zfcp-2.8-u2-nh.cnf.primal & 2.33 & \numprint{32697150} & \numprint{65} & \numprint{1777221} & \numprint{10950109} \\
ss & 21.03 & \numprint{1652680} & \numprint{43} & \numprint{27} & \numprint{1652680} \\
stokes & 30.51 & \numprint{11449533} & \numprint{720} & \numprint{1705} & \numprint{11449533} \\
uk-2002 & 18.92 & \numprint{15759370} & \numprint{2450} & \numprint{194942} & \numprint{18520486} \\
uk-2005 & 26.71 & \numprint{35051939} & \numprint{5213} & \numprint{1776852} & \numprint{39459925} \\
vas\_stokes\_1M & 31.88 & \numprint{1090664} & \numprint{577} & \numprint{1013} & \numprint{1090664} \\
vas\_stokes\_2M & 30.34 & \numprint{2146677} & \numprint{637} & \numprint{1284} & \numprint{2146677} \\
vas\_stokes\_4M & 30.03 & \numprint{4382246} & \numprint{613} & \numprint{1116} & \numprint{4382246} \\
wb-edu & 8.26 & \numprint{6920306} & \numprint{3841} & \numprint{25762} & \numprint{9845725} \\

  \end{longtblr}

\subsubsection{Graph instances}
\newcolumntype{b}{>{\ttfamily}X}
\newcolumntype{s}{>{\hsize=.15\hsize\raggedleft\arraybackslash}X}
\begin{tabularx}{\textwidth}{bss}
\caption{Benchmark set of graphs with $|V|$ representing the amount of nodes and $|E|$ the number of edges.}
\label{tab:graphinstances}
\endfirsthead

\caption[]{(Continuation)} \\
\toprule
\endhead

\bottomrule
\multicolumn{3}{r}{\textit{Continuation on the next page}} \\
\endfoot

\bottomrule
\endlastfoot

\toprule
\textbf{Graph} & \textbf{$|V|$} & \textbf{$|E|$} \\
\midrule
af\_shell10 & \numprint{1508065} & \numprint{25582130} \\
af\_shell9 & \numprint{504855} & \numprint{8542010} \\
asia.osm & \numprint{11950757} & \numprint{12711603} \\
audikw1 & \numprint{943695} & \numprint{38354076} \\
belgium.osm & \numprint{1441295} & \numprint{1549970} \\
cage15 & \numprint{5154859} & \numprint{47022346} \\
citationCiteseer & \numprint{268495} & \numprint{1156647} \\
coAuthorsCiteseer & \numprint{227320} & \numprint{814134} \\
coAuthorsDBLP & \numprint{299067} & \numprint{977676} \\
coPapersCiteseer & \numprint{434102} & \numprint{16036720} \\
coPapersDBLP & \numprint{540486} & \numprint{15245729} \\
delaunay\_n19 & \numprint{524288} & \numprint{1572823} \\
delaunay\_n20 & \numprint{1048576} & \numprint{3145686} \\
delaunay\_n21 & \numprint{2097152} & \numprint{6291408} \\
delaunay\_n22 & \numprint{4194304} & \numprint{12582869} \\
delaunay\_n23 & \numprint{8388608} & \numprint{25165784} \\
delaunay\_n24 & \numprint{16777216} & \numprint{50331601} \\
ecology1 & \numprint{1000000} & \numprint{1998000} \\
ecology2 & \numprint{999999} & \numprint{1997996} \\
G3\_circuit & \numprint{1585478} & \numprint{3037674} \\
germany.osm & \numprint{11548845} & \numprint{12369181} \\
great-britain.osm & \numprint{7733822} & \numprint{8156517} \\
italy.osm & \numprint{6686493} & \numprint{7013978} \\
kkt\_power & \numprint{2063494} & \numprint{6482320} \\
ldoor & \numprint{952203} & \numprint{22785136} \\
luxembourg.osm & \numprint{114599} & \numprint{119666} \\
netherlands.osm & \numprint{2216688} & \numprint{2441238} \\
nlpkkt120 & \numprint{3542400} & \numprint{46651696} \\
nlpkkt160 & \numprint{8345600} & \numprint{110586256} \\
nlpkkt200 & \numprint{16240000} & \numprint{215992816} \\
nlpkkt240 & \numprint{27993600} & \numprint{373239376} \\
rgg\_n\_2\_15\_s0 & \numprint{32768} & \numprint{160240} \\
rgg\_n\_2\_16\_s0 & \numprint{65536} & \numprint{342127} \\
rgg\_n\_2\_17\_s0 & \numprint{131072} & \numprint{728753} \\
rgg\_n\_2\_18\_s0 & \numprint{262144} & \numprint{1547283} \\
rgg\_n\_2\_19\_s0 & \numprint{524288} & \numprint{3269766} \\
rgg\_n\_2\_20\_s0 & \numprint{1048576} & \numprint{6891620} \\
rgg\_n\_2\_21\_s0 & \numprint{2097152} & \numprint{14487995} \\
rgg\_n\_2\_22\_s0 & \numprint{4194304} & \numprint{30359198} \\
rgg\_n\_2\_23\_s0 & \numprint{8388608} & \numprint{63501393} \\
rgg\_n\_2\_24\_s0 & \numprint{16777216} & \numprint{132557200} \\
thermal2 & \numprint{1227087} & \numprint{3676134} \\
\end{tabularx}

\end{document}